\documentclass[runningheads]{llncs}

\usepackage{amsmath}
\usepackage{amsfonts}
\usepackage{amssymb}
\usepackage{booktabs}
\usepackage{graphicx}
\usepackage[colorlinks]{hyperref}
\usepackage[switch]{lineno}
\usepackage[shortlabels]{enumitem}
\usepackage{subcaption}
\usepackage{longtable}

\usepackage[dvipsnames]{xcolor}

\newcommand{\cand}{\mathcal{C}}

\newcommand{\assertion}{\mathcal{A}}

\newcommand{\ballots}{\mathcal{B}}
\newcommand{\ballot}{\beta}
\newcommand{\election}{\mathcal{L}}

\newcommand{\seats}{S}

\newcommand{\quota}{Q}
\newcommand{\proj}{\sigma}

\DeclareMathOperator{\first}{first}

\DeclareMathOperator{\argmin}{argmin}

\newcommand{\NWA}{\mathit{NWA}}
\newcommand{\NEB}{\textsf{AG}}
\newcommand{\IQ}{\textsf{IQ}}
\newcommand{\UT}{\textsf{UT}}

\newcommand{\CNEB}{\textsf{NL}}

\newcommand{\Lbetter}{L_\textrm{elim}}

\newcommand{\ignore}[1]{}

\title{A First Approach to Risk-Limiting Audits for Single Transferable Vote
Elections\thanks{To appear in the Workshop on Advances in Secure Electornic
Voting (Voting'22), on 18 February 2022.}}

\titlerunning{A First Approach to Risk-Limiting Audits for STV Elections}

\author{
Michelle Blom     \inst{1}   \orcidID{0000-0002-0459-9917}  \and
Peter J. Stuckey  \inst{2}   \orcidID{0000-0003-2186-0459}  \and
Vanessa Teague    \inst{3}   \orcidID{0000-0003-2648-2565}  \and
Damjan Vukcevic   \inst{4,5} \orcidID{0000-0001-7780-9586}}

\authorrunning{Blom, Stuckey, Teague, Vukcevic}
\institute{
School of Computing and Information Systems, University of Melbourne,
Parkville, Australia \\
\email{michelle.blom@unimelb.edu.au}
\and
Department of Data Science and AI, Monash University, Clayton, Australia
\and
Thinking Cybersecurity Pty.\ Ltd., Melbourne, Australia
\and
School of Mathematics and Statistics, University of Melbourne, Parkville,
Australia
\and
Melbourne Integrative Genomics, University of Melbourne, Parkville,
Australia \\
\email{damjan.vukcevic@unimelb.edu.au}}

\begin{document}

\maketitle

\begin{abstract}
Risk-limiting audits (RLAs) are an increasingly important method for checking
that the reported outcome of an election is, in fact, correct. Indeed, their
use is increasingly being legislated. While effective methods for RLAs have
been developed for many forms of election---for example: first-past-the-post,
instant-runoff voting, and D'Hondt elections---auditing methods for single
transferable vote (STV) elections have yet to be developed. STV elections are
notoriously hard to reason about since there is a complex interaction of votes
that change their value throughout the process.  In this paper we present the
first approach to risk-limiting audits for STV elections, restricted to the
case of 2-seat STV elections.
\end{abstract}

% ---------------------------------------------------------------------------

\section{Introduction}

Single transferable vote (STV) elections are a method for selecting candidates
to fill a set of seats in a single election, which tries to achieve
proportional representation with respect to voters' preferences expressed as a
ranked list of candidates.  STV elections are used in many places throughout
the world including Australian Senate elections, all elections in Malta,
provincial elections in Canada, many elections in Ireland, and in more than 20
cities in the USA. STV elections are considered as one of the better multi-seat
election methods because they achieve some form of ranked proportional
representation, unlike many multi-seat elections, although some consider the
complexity for voters of having to rank candidates a drawback.

STV elections are one of the most complex form of election to reason about
because the value of ballots can change across the election process.  When a
candidate achieves a tally of votes large enough to be awarded a seat (a
\emph{quota}) then each ballot currently in their tally is transferred to the
next eligible candidate listed on the ballot, at a reduced value (the
\emph{transfer value}). The transfer value is calculated (and there are a
number of possibilities here) so the total value of the ballots transferred is
no greater than the tally minus the quota, thus enforcing the idea that each
vote has a value of 1 which may be used (in parts) in electing multiple
candidates.

Risk-limiting audits (RLAs)~\cite{lindemanStark12} are a form of auditing of
election results to determine with some statistical likelihood that the correct
result was determined. They rely on comparing paper ballots, the ground truth
of the election, with the electronic recorded information to check the result.
The \emph{risk limit} is an upper bound on the probability that an incorrect
election outcome will not be corrected by the audit. RLAs are increasingly used
around the world, and sometimes their use is mandated by legislation. While RLA
methods have been determined for many forms of elections:
first-past-the-post~\cite{lindemanStark12}, any scoring function,\footnote{%
Any social choice function that is a \emph{scoring rule}---that assigns
`points' to candidates on each ballot, sums the points across ballots, and
declares the winner(s) to be the candidate(s) with the most `points'---can be
audited using SHANGRLA (see below).} instant-runoff voting
(IRV)~\cite{blom2019raire}, D'Hondt~\cite{StarkTeague2014} and Hamiltonian
elections~\cite{voting21}, there are currently no approaches to risk-limiting
audits for STV elections.  In this paper we make a first step towards this,
restricting attention to 2-seat STV elections, which are the simplest form.  To
do so we generate auditing machinery which should also be useful for larger STV
elections, but we leave the exact mechanisms required  as future work.

% ---------------------------------------------------------------------------

\section{Preliminaries}

\subsection{Single transferable vote elections}

STV is a multi-winner preferential voting system.  Voters rank candidates (or
parties) in order of preference. The $\seats$ seats are allocated in a way that
reflects both \emph{proportionality} (voting blocks should be represented in
approximately the proportions that people vote for them) and \emph{preference}
(if a voter's favourite candidate cannot win, or receives more than necessary
for a seat, that voter's later preferences influence who else gets a seat).

The set of candidates is $\cand.$ A ballot $\ballot$ is a sequence of
candidates $\pi$, listed in order of preference (most popular first), without
duplicates but without necessarily including all candidates. We use list
notation (e.g., $\pi = [c_1,c_2,c_3,c_4]$). The notation $\first(\pi) = \pi(1)$
denotes the first item (candidate) in  sequence $\pi$. An STV election
$\election$ is defined as a multiset\footnote{A multiset allows for the
inclusion of duplicate items.} of ballots.

\begin{definition}[STV Election]
An STV election $\election$ is a tuple $\election = (\cand, \ballots, \quota,
\seats)$ where $\cand$ is a set of candidates,  $\ballots$ the multiset of
ballots cast, $\quota$ the election quota (the number of votes a candidate must
attain to win a seat---usually the Droop quota---\autoref{eqn:Droop}), and
$\seats$ the number of seats to be filled.
\begin{equation}
\quota = \left\lfloor \frac{|\ballots|}{\seats + 1} \right\rfloor + 1
\label{eqn:Droop}
\end{equation}
\end{definition}

\begin{definition}{\textbf{Projection} $\mathbf{\proj_\mathcal{S}(\pi)}$}
We define the projection of a sequence $\pi$ onto a set $\mathcal{S}$ as the
largest subsequence of $\pi$ that contains only elements of $\mathcal{S}$. (The
elements keep their relative order in $\pi$.) For example:
$\proj_{\{c_2,c_3\}}([c_1,c_2,c_4,c_3]) = [c_2,c_3]$ 
 and $\proj_{\{c_2,c_3,c_4,c_5\}}([c_6,c_4,c_7,c_2,c_1]) = [c_4,c_2].$
\label{def:Projection}
\end{definition}

Each ballot starts with a value of 1, and may change its value as counting
progresses.  Throughout the count, each eligible candidate has a non-decreasing
\textit{tally} of ballots.  Ballots can be redistributed between candidates in
two ways.  If a candidate achieves a quota, their ballots will be redistributed
with a reduced value. If a candidate is \emph{eliminated}, their ballots are
passed down the preference list at their current value. The following
paragraphs describe the algorithm.

Initially, each ballot's value is 1 and each candidate is awarded all ballots
on which they are ranked first.  A seat is awarded to every candidate whose
tally has reached or exceeded $\quota$. When candidate $c \in \cand$  achieves
a quota, the ballots counting towards their tally  are distributed to remaining
eligible candidates at a reduced value as follows. (A candidate is eligible if
they have not been eliminated, and their tally has not reached a quota's worth
of votes.) Let $V_c$ denote the total \emph{value} of ballots counting towards
$c$ in the round that $c$ is awarded a seat, and $|\ballots_c|$ the
\emph{number} of those ballots. Each of these ballots is given a new value of
$\tau$, and distributed to the next most preferred eligible candidate on the
ballot.

One way of computing $\tau$ is the \emph{unweighted Gregory method}, given by:
\begin{equation}
\tau = \frac{V_c - \quota}{|\ballots_c|}.
\label{eqn:tvalue}
\end{equation}
This method is used in Australian Senate elections.  Note that ballots can
increase in value after a second transfer, but never above 1.
There are alternative ways to calculate transfer value, but our analysis is
agnostic about them, as long as they satisfy some bounds described in
\autoref{subsec:improvedUpperBound}.  Our empirical results use the unweighted
Gregory method, but other methods are likely to be very similar. We do not
consider randomised methods for distributing votes.

If no candidate has achieved a quota, the candidate $c_e$ with the fewest votes
is eliminated. Each ballot currently counting towards $c_e$ is distributed to
its next most preferred eligible candidate, at its current value.

Each round of counting thus either awards seats to candidates that have
achieved a quota, or eliminates a candidate with the lowest tally.  Either way,
their ballots are redistributed. This continues until either all seats have
been awarded, or the number of eligible candidates equals the number of seats
left to be awarded. In the latter case, every remaining candidate is awarded a
seat.

\emph{Terminology:}  We will use the term ``is seated'' to include either way
of getting a seat, while ``gets a quota'' is reserved for getting a seat by
obtaining a quota.  We say a candidate is ``eligible'' if it has not been
eliminated nor reached a quota.

\begin{table}
\begin{subtable}{.3\columnwidth}
\centering
\begin{tabular}{lr}
& \\
\toprule
Ranking & Count \\
\midrule
{}[$c_1$, $c_3$]        &  8,001 \\
{}[$c_1$]               &  1,000 \\
{}[$c_2$, $c_3$, $c_4$] &  3,000 \\
{}[$c_3$, $c_4$]        &  5,000 \\
{}[$c_4$, $c_1$, $c_2$] &  4,000 \\
\midrule
Total                   & 21,001 \\
\bottomrule
\end{tabular}
\caption{}
\label{tab:EGSTV1a}
\end{subtable}$\,\,\,$
\begin{subtable}{.7\columnwidth}
\centering
\begin{tabular}{crrr}
\multicolumn{4}{c}{Seats: 2 \quad Ballots: 21,001 \quad Quota: 7,001} \\
\toprule
Candidate & Round 1 & Round 2 & \phantom{M} Round 3 \\
\midrule
 & Elect $c_1$       & \phantom{M} Eliminate $c_2$ & Elect $c_3$ \\
 & $\tau_1 = 0.2222$ &                             &             \\
\midrule
  $c_1$ &  9,001  &  ---   &  ---    \\
  $c_2$ &  3,000  &  3,000 &  ---    \\
  $c_3$ &  5,000  &  6,778 &  9,778  \\
  $c_4$ &  4,000  &  4,000 &  4,000  \\
\midrule
 Total  & 21,001  & 13,778 & 13,778  \\
\bottomrule
\end{tabular}
\caption{}
\label{tab:EGSTV1b}
\end{subtable}
\caption{(\autoref{ex:ex}) An STV election profile, stating (a)~the number of
ballots cast with each listed ranking over candidates $c_1$ to $c_4$, and
(b)~the tallies after each round of counting, election, and elimination.}
\vspace{-0.4cm}
\label{tab:EGSTV1}
\end{table}

\begin{example}\label{ex:ex}
Consider a 2-seat STV election with four candidates $\cand = \{ c_1, c_2, c_3,
c_4\}$ with  ballots $\ballots$ shown in \autoref{tab:EGSTV1a}. The (Droop)
quota for this election is calculated as $Q = \left\lfloor 21001 / 3
\right\rfloor + 1 = 7001$.  The election proceeds as shown in
\autoref{tab:EGSTV1b}.  Candidate $c_1$ initially has more than a quota and is
elected to a seat.

Using the unweighted Gregory method, the transfer value $\tau$ is determined as
$2000 / 9000 = 0.2222$.  The 8001 transferable ballots with ranking [$c_1,$
$c_3$] go to $c_3$ each with value 0.2222 for a total of 1778. The remaining
ballots in $c_1$'s tally have a ranking of [$c_1$].  These ballots have no
eligible next preference and are exhausted (not redistributed).  Note how some
vote value 222 is lost here.

In the next round no candidate has a quota so the candidate $c_2$ with the
least tally is eliminated. The votes in their pile all flow to $c_3$ as next
remaining unelected candidate.  Now $c_3$ has a quota and is elected.
\qed
\end{example}

\subsection{Assertion-based risk-limiting audits}

SHANGRLA~\cite{shangrla} is a general framework for conducting RLAs. It offers
a wide variety of social choice functions, statistical risk functions and audit
designs (such as stratified audits or ballot-comparison audits).

This generality is achieved by abstraction: a SHANGRLA audit first reduces the
correctness of a reported outcome to the truth of a set $\mathcal{A}$ of
quantitative \emph{assertions} about the set of validly cast ballots, which can
then be tested using statistical methods. The assertions are either true or
false depending on the votes on the ballots. If every assertion in
$\mathcal{A}$ is true, the reported outcome is correct.  $\mathcal{A}$
generally depends on the    social choice function and the reported electoral
outcome, and may also depend on the cast vote records (CVRs), vote subtotals,
or other data generated by the voting system.

For example, in a first-past-the-post election in which Alice is the apparent
winner, $\mathcal{A}$ could include an assertion, for each other candidate $c$,
that there are more  votes for Alice than $c.$ In this example, $\mathcal{A}$
is both necessary and sufficient: if any assertion $A \in \mathcal{A}$ is
false, then Alice did not win (except possibly in a tie).  In general, however,
the assertions in $\mathcal{A}$ must be \emph{sufficient} to imply that the
announced election outcome is correct, but they need not be necessary: the
announced electoral result may be correct even if some assertions in
$\mathcal{A}$ are false.  The assertions we derive for STV in this paper are
sufficient but not necessary for supporting the announced election outcome. 

SHANGRLA expresses each assertion $A \in \mathcal{A}$  as an \emph{assorter},
which is a function that assigns a nonnegative value to each ballot, depending
on the selections the voter made on the ballot and possibly other information
(e.g.\ reported vote totals or CVRs). The assertion is true iff the mean of the
assorter (over all ballots) is greater than $1/2$. Generally, ballots that
support the assertion score higher than $1/2$, ballots that weigh against it
score less than $1/2$, and neutral ballots score exactly $1/2$. In the
first-past-the-post example above, $A$ might assert that Alice's tally is
higher than Bob's. The corresponding assorter would assign 1 to a ballot if it
has a vote for Alice, 0 if it has a vote for Bob, and $1/2$ if it has a no
valid vote for either.

% ---------------------------------------------------------------------------

\section{Reasoning about STV elections: deriving bounds and assertions}
\label{sec:reasoning}

In order to make verifiable assertions about STV elections we need to examine
how we can reason about STV elections. In this section we define testable
assertions for reasoning about STV elections (summarised in
\autoref{tab:definitions}).

\begin{table}[t]
\centering
\caption{Summary of definitions.}
\begin{tabular}{llr}
\toprule
Quantity/Assertion & Description & Page  \\
\midrule
\multicolumn{3}{l}{\textbf{Lower and upper bounds}} \\
$L_\textrm{basic}(c)$ &
    First preferences for $c$ &
    \pageref{quantity:lower-basic} \\
$U_\textrm{basic}(c)$ &
    Ballots mentioning $c$ &
    \pageref{quantity:upper-basic} \\
$U_\textrm{comp}(c,c')$ &
    Ballots where $c$ appears before $c'$ &
    \pageref{quantity:upper-comparative} \\
$\Lbetter(w,O)$ &
    Lower bound for $w$'s tally, assuming it is never &
    \pageref{quantity:lower-better} \\
  & \quad less than that of each candidate in $O$ \\
$U_\textrm{complex}(c,b,W,\overline{\tau})$ \phantom{M} &
    A complex upper bound for $c$'s tally &
    \pageref{quantity:upper-complex}  \\
\addlinespace
\multicolumn{3}{l}{\textbf{Assertions}} \\
$\IQ(c)$ &
    $c$ gets a quota initially &
    \pageref{assertion:IQ} \\
$\UT(c,\overline{\tau})$ &
    $c$'s transfer value is less than $\overline{\tau}$ &
    \pageref{assertion:UT} \\
$\NEB(w,l)$ &
    $w$'s tally is always greater than $l$'s tally &
    \pageref{assertion:NEB} \\
$\CNEB(w,l,W,$ $\overline{\tau},G,O)$ &
    $w$ `never loses' to $l$, given some assumptions &
    \pageref{assertion:CNEB} \\
\bottomrule
\end{tabular}
\label{tab:definitions}
\end{table}

\subsection{Simple bounds and assertions}

A simple lower bound on the tally of candidate $c$ is the number of first
preference votes they receive: $L_\textrm{basic}(c) = | \{ \ballot : \ballot
\in \ballots, \first(\ballot) = c \}|$. \label{quantity:lower-basic}

Given this bound we can introduce our first type of assertion, that a candidate
gets a quota initially: $\IQ(c) \equiv L_\textrm{basic}(c) \geqslant \quota$.
\label{assertion:IQ}

\begin{lemma}
If $\IQ(c)$ holds then $c$ is seated.
\qed
\end{lemma}

The next assertion we introduce is one that upper bounds the transfer value at
$\overline{\tau}$ for candidates that have an initial quota:
$\UT(c,\overline{\tau}) \equiv L_\textrm{basic}(c) < \quota / (1 -
\overline{\tau})$. \label{assertion:UT}
Clearly if the initial tally for $c$ is $T < \quota / (1 - \overline{\tau})$
then the transfer value (using the unweighted Gregory method) is $(T - \quota)
/ T < \overline{\tau}$.  We use this in \autoref{sec:OneQuota} to improve upper
bounds on tallies.

A simple upper bound on the tally of a candidate $c$ is the number of ballots
on which they appear: $U_\textrm{basic}(c) = | \{ \ballot : \ballot \in
\ballots, c \text{ occurs in } \ballot \}|$. \label{quantity:upper-basic}

We can improve this upper bound when comparing against an alternative candidate
$c'$.  The number of ballots where $c$ appears before $c'$ (including the case
where $c'$ doesn't appear) is $U_\textrm{comp}(c,c') = | \{ \ballot : \ballot
\in \ballots, \first\left(\proj_{\{c,c'\}}(\ballot)\right) = c \}|$.
\label{quantity:upper-comparative}
This is the maximum number of ballots that can appear in the tally of $c$
before $c'$ is eliminated.

We can use this to state a sufficient condition that candidate $w$'s tally is
\emph{always greater}\footnote{Previous IRV auditing work~\cite{blom2019raire}
has used the term \emph{not eliminated before} for this concept, but we reserve
it for a more restrictive notion defined below.} than candidate $l$'s:
$\NEB(w,l) \equiv L_\textrm{basic}(w) > U_\textrm{comp}(l,w)$.
\label{assertion:NEB}

\begin{lemma}
If $\NEB(w,l)$ holds then candidate $w$'s tally is always greater than $l$'s.
\end{lemma}
\begin{proof}
Candidate $w$ always has a tally of at least $L_\textrm{basic}(w)$. Candidate
$l$ always has a tally of at most $U_\textrm{comp}(l,w)$ while $w$ is not
eliminated nor seated. Also, by assumption, $L_\textrm{basic}(w) >
U_\textrm{comp}(l,w)$, which means $w$'s tally always exceeds $l$'s tally while
$w$ is not eliminated nor seated.
\qed
\end{proof}

\begin{corollary}
\label{coro:NEB}
If $\NEB(w,l)$ holds then $l$ cannot be seated when $w$ is not.
\end{corollary}
\begin{proof}
$\NEB(w,l)$ implies we cannot eliminate $w$ before $l$.
\qed
\end{proof}

\subsection{Improving the lower bound} \label{subsec:improvedLowerBound}

Note that the $\NEB$ condition is very strong---there are many cases where
candidate $w$ does not lose to $l$ but $\NEB(w,l)$ does not hold. We can
improve this by using the knowledge of easily proven $\NEB$ conditions to
improve lower bounds on the tally of $w$ at any point at which $w$ could be
eliminated. At such points, we know that any candidate $o \in O$ for which
$\NEB(w,o)$ holds must have already been eliminated. Any ballots that would
move from $o$ to $w$ on the elimination of $o$ can be counted towards this
lower bound.  That motivates the following definition for an improved lower
bound:
\[ \Lbetter(w,O) = | \{ \ballot : \ballot \in \ballots, \first(\proj_{\cand -
    O}(\ballot)) = w \}|. \label{quantity:lower-better} \]
This allows us to reason about when $w$ might be eliminated, in particular, and
to prove that it cannot be.

\begin{lemma}\label{lemma:ilb}
Given a candidate $w$ and a set of candidates $O$, suppose $\NEB(w,o)$ holds
for all $o \in O$. Then $\Lbetter(w,O)$ is a lower bound on $w$'s tally at any
point at which it could be eliminated.
\end{lemma}
\begin{proof}
By assumption, $w$ cannot be eliminated before any candidate in $O$. If any
candidate in $O$ is seated, $\NEB(w,o)$ implies that $w$ must also be seated
(\autoref{coro:NEB}). So at any point at which $w$ could be eliminated,  all
candidates in $O$ are eliminated. Hence all the ballots in  $| \{ \ballot :
\ballot \in \ballots, \first(\proj_{\cand - O}(\ballot)) = w \}|$ contribute to
$w$'s tally. Since none of the candidates in $O$ reached a quota, all the
ballots still have their full value.
\qed
\end{proof}

\subsection{Improving the upper bound} \label{subsec:improvedUpperBound}

The simple $\NEB$ condition will fail when a candidate appears in many ballots,
but the values of these ballots are ``used up'' by seating earlier candidates.
In the following example, $\NEB(c_4,c_2)$ does not hold, but more careful
reasoning allows us to prove that $c_4$ cannot lose to $c_2$.

\begin{example}\label{ex:hard}
Consider a 2-seat election with ballots and multiplicities defined as
$[c_1,c_2]\colon 30$, $[c_4,c_1,c_2]\colon 20$, $[c_3,c_1,c_2]\colon 4$,
$[c_2,c_4]\colon 2$, $[c_3]\colon 4$, where the quota is 21. We cannot show
$\NEB(c_4,c_2)$ since $c_2$ appears in 30 ballots with $c_1$.  The maximum
transfer value in a 2-seat election is 2/3 (which can only occur if one
candidate gets all the ballots initially).  If we note that $c_1$ must be
seated, we can see that the maximum value $c_2$ can derive from these ballots
is 20. This still makes it impossible to show $c_4$ cannot lose to $c_2$.  In
fact the actual transfer value is 0.3, and with this $c_2$ can only gather 9.
Using a maximal transfer value of $0.3$ we could show that $c_4$ cannot lose to
$c_2$.
We have to be careful to consider the ballots $[c_3,c_1,c_2]$; since these are
not in $c_1$'s pile when it obtains a quota, they are not reduced in value.
When $c_3$ is eliminated they are passed to $c_2$ (since $c_1$ is already
seated) at full value.
\qed
\end{example}

In order to more effectively upper bound the tally of a candidate, we need to
reason about the possible transfer values of ballots that follow this route. 

We have the following trivial upper bound on transfer values: the maximum
transfer value $\tau$ in an $S$-seat election (using the unweighted Gregory
method) is $\tau = S/(S+1)$. This is only possible if one candidate gains all
the votes initially.

We now define a complex bound that relies on a number of assumptions.  We are
trying to find an upper bound on the tally of some candidate $c$ in order to
compare them with an alternate winner $b.$

Assume that all candidates in $W$ are seated (which may happen before, during
or after this bound is computed, and may occur by getting a quota or by
remaining at the end). The only candidates who may be seated but are not in $W$
are $b$ and $c$.  Let $\overline{\tau}$ be a vector of upper bounds
$\overline{\tau}_w, w \in W$, that is the maximum transfer value for any ballot
that was in $w$'s pile at the time it was seated (if it was).\footnote{We
simply require that an upper bound on a ballot's value is the maximum of the
upper bounds on per-candidate transfer values. This is true of transfer values
calculated according to the unweighted Gregory method, and weighted methods.}
Candidates in $W$ clearly cannot be eliminated, they are either eligible or
seated.  Assume also that $b$ is eligible.

Let $R$ be all of the other candidates, $R = \cand - W - \{b,c\}$.
Let $G$ be candidates for which $\NEB(g,c)$ hold for all $g \in G$.
Under these assumptions we define an upper bound on the tally of candidate $c$
as follows: 
\begin{equation}
\label{quantity:upper-complex}
 U_\textrm{complex}(c,b,W,\overline{\tau},G) =
 \sum_{\ballot \in \ballots} U_\textrm{complex}(c,b,W,\overline{\tau},\ballot)
\end{equation}
where
\begin{equation*}
U_\textrm{complex}(c,b,W,\overline{\tau},G,\ballot) = 
\left\{ 
\begin{array}{ll}
   0 & \exists g \in G - W \text{~s.t.~}\first(\proj_{g,c}(\ballot)) = g\\
   0 & c \text{ does not occur in } \beta \\
   0 & \first(\proj_{b,c}(\ballot)) = b\\
mt_w & \first(\ballot) \in W \\
   1 & \text{otherwise}
\end{array}
\right.
\end{equation*}
where $mt_w = \max\{\overline{\tau}_w : w \in W \text{~precedes~} c \text{~in~}
\ballot\}$.

\begin{lemma}\label{lemma:iub}
Under the assumption that only candidates $W \cup \{b,c\}$ can be seated, and
that all candidates in $W$ are seated (though this may happen before, during or
after this comparison), with upper bound on the transfer values
$\overline{\tau}$, and $b \not\in W$ is eligible, and that $\NEB(g,c)$ holds
for all $g \in G$, then $U_\textrm{complex}(c,b,W,\overline{\tau},G)$ is an
upper bound on the tally of $c$.
\end{lemma}
\begin{proof}
Consider each ballot $\ballot$, and each case in the definition of
$U_\textrm{complex}$.

If there exists, before $c$, on $\ballot$, a candidate in $ g \in G - W$ then
$c$ is preceded on the ballot by a candidate who cannot be eliminated before
$c$ and, by assumption, cannot win. Hence $\ballot$ counts 0 towards $c$'s
tally.

If $c$ does not occur in $\ballot$, or $b$ precedes $c$ in $\ballot$, then
clearly the ballot contributes 0 to $c$'s tally (given the assumption that $b$
is eligible).

If the first candidate on the ballot is $w \in W$ then we need to consider
whether or not $w$ has been seated. If $w$ is unseated, then the ballot still
sits with them and it contributes 0 to $c$'s tally. If $w$ is seated in the
last round without a quota, then it contributes 0 to $c$'s tally.  If $w$ has
obtained a quota, then the ballot has definitely been transferred at least
once.  If it ends up in $c$'s pile, it can only have been involved with
transfers that appear before $c$. The maximum value it can have is the maximum
of the transfer values. This is the overall maximum (since the others are
zero.)

The remaining ballots have maximum possible value 1.  Note that the case where
$\first(\proj_{W\cup R}(\ballot)) \in W$ does not imply that a vote has been
transferred if it reaches $c$'s pile. It may have been that the winner $w$ was
seated before the ballot reached $w$'s pile, in which case it could arrive in
$c$'s pile by elimination rather than by transfer.

Since each ballot is counted at its maximum possible value given the
assumptions, the upper bound is correct.
\qed
\end{proof}

We can now define a refined version of the `always greater' assertion that
takes into account this new bound. We define $w$ \emph{never loses} to $l$,
denoted $\CNEB(w,l,W,$ $\overline{\tau},G,O)$ as follows.
\label{assertion:CNEB}
Assume any previous winners are included in the set $W$ with upper bounds on
transfer values $\overline{\tau}$.
Assume $\NEB(w,o)$ holds for $o \in O$ and $\NEB(g,l)$ holds for all $g \in G$.
Then
\[ \CNEB(w,l,W,\overline{\tau},G,O) \equiv
   \Lbetter(w,O) > U_\textrm{complex}(l,w, W,\overline{\tau},G). \]

\begin{lemma}\label{lemma:wlo}
Under the assumption that only candidates $W \cup \{w,l\}$ can be seated, and
that all candidates in $W$ are seated (though this may happen before, during or
after this comparison), with maximal transfer values $\overline{\tau}$, and
$\NEB(w,o)$ holds for all $o \in O$ then and $\NEB(g,l)$ holds for all $g \in
G$ then $\CNEB(w,l,W,\overline{\tau},G,O)$ implies that $w$ never loses to
$l$.
\end{lemma}
\begin{proof}
Suppose to the contrary that we are about to eliminate $w$.  By
\autoref{lemma:ilb} its tally is at least $\Lbetter(w,O)$. In order for $l$ to
get a seat it cannot already be eliminated. By assumption neither can any of
$W$.  By \autoref{lemma:iub}, $U_\textrm{complex}(l,w,W,\overline{\tau},G)$ is
an upper bound on the tally of $l$, since we have treated all other candidates
as eliminated.  Before $w$ is eliminated $l$ can never have more tally than
this bound.  Because the tally of $w$ is greater than the tally of $l$ it
cannot be eliminated.

From the above, $w$ can never be eliminated, therefore every candidate who is
seated must obtain a quota (otherwise $w$ would be seated at the end). Assume
$w$ is not seated. Now none of the ballots in $\Lbetter(w,O)$ can ever sit with
any candidate that is seated, since none of $O$ can be seated if $w$ is not.
Then the total tally of the $\seats$ winners is at least $\seats \times
\quota$, and none of the ballots in $\Lbetter(w,O)$ are included. Suppose to
the contrary $l$ is a winner, then its maximum quota when elected is
$U_\textrm{complex}(l,w,W,\overline{\tau},G)$ which is, by the $\CNEB$
assumption, less than $\Lbetter(w,O)$. Hence $\Lbetter(w,O) > \quota$. But this
gives a total tally of votes greater than $(\seats + 1) \times \quota$, more
than exist in the entire election. Contradiction.
\qed
\end{proof}

% ---------------------------------------------------------------------------

\section{Deriving assorters} \label{sec:assorters}

To use the SHANGRLA framework, we need to determine an assorter for each
assertion defined in \autoref{sec:reasoning}.  It suffices to show how to write
each one as a \emph{linear assertion} as per the general framework described by
Blom et al.\ \cite{evote2021assertions}.

Assertions $\NEB$ and $\CNEB$ involve comparing two tallies.  These can be
straightforwardly written in the standard linear form.

Assertion $\IQ$ is of the form $T \geqslant Q$ and assertion $\UT$ is of the
form $T < a \cdot Q$, for some tally $T$ and positive constant $a$.  These are
not immediately in linear form because $Q$ is not a tally nor a simple function
of a tally.  However, for each of these we can define a linear assertion that
is either equivalent or stricter.

To get an assertion of the form $T \geqslant Q$, we instead work with the
assertion $T > |\ballots| / (\seats + 1)$.  This latter assertion can be
written in linear form since $|\ballots|$ is a tally (simply count each valid
ballot).  To see that this implies our desired assertion, consider that $T >
|\ballots| / (\seats + 1) \geqslant \left\lfloor|\ballots| / (\seats +
1)\right\rfloor$.  Since $T$ is an integer strictly greater than the term on
the right, which is itself an integer, it must be at least 1 greater than that
term.  That is, $T \geqslant \left\lfloor|\ballots| / (\seats + 1)\right\rfloor
+ 1 = Q$.

To get an assertion of the form $T < a \cdot Q$, we instead work with the
slightly stricter assertion $T < a \cdot |\ballots| / (\seats + 1)$, which is
clearly expressible in linear form.  The floor function has the property that
$\lfloor x\rfloor \leqslant x < \lfloor x\rfloor + 1$.  Taking the right-hand
part of this double inequality and setting $x = |\ballots| / (\seats + 1)$
gives $|\ballots| / (\seats + 1) < Q$.  Our working assertion therefore implies
our desired assertion, $T < a \cdot Q$.

% ---------------------------------------------------------------------------

\section{RLAs for 2-seat STV elections}\label{sec:Methods}

Given the assertions we have defined in the previous section, we are now ready
to define an algorithm to choose a set of assertions that, if validated, will
ensure, within the risk limit $\alpha$, that the election result must be
correct.  We will try to choose a set of assertions that is expected to be
auditable by viewing few ballots.  We assume a function
$ASN(a,\alpha,\epsilon)$ that returns the \emph{average sample number} for
verifying assertion $a$, that is the expected number of ballots required to
verify the assertion $a$ if it indeed holds, given the recorded election data,
a risk limit $\alpha$ and expected error rate $\epsilon$.\footnote{The
\emph{expected error rate} is the expected proportion of ballots that are
counted erroneously when calculating the assorter corresponding to assertion
$a$.} For some assertions there are closed-form formulae for this estimation,
but in general we can use sampling to provide accurate estimates.  Note that
the expected auditing effort is not relevant to proving that the assertions, if
verified, certify the election result up to risk limit $\alpha$. Rather, we use
it to suggest a set of assertions that are expected to be easy to audit.

Assume the declared winners of the election $\election$ are $DW = \{w_1,w_2\}$.
We need to consider all possible alternative election results $AR =
\{\{c_1,c_2\} : \{c_1,c_2\} \subseteq \cand, \{c_1,c_2\} \neq DW\}$, and verify
assertions that will invalidate all such results.

We first use simple $\NEB$ assertions to eliminate as many pairs as possible.
\textsf{NonWinners} (\autoref{fig:nw-short}) finds a set of candidates, denoted
$NW$, that \textit{clearly cannot win}. In \textsf{NonWinners}, we first
determine all the \textit{always greater} relationships $\NEB$ that hold on the
basis of the recorded election result (lines 2--4). We then find the
candidates $c$ for which there exists at least two other candidates $w_1$ and
$w_2$ such that $\NEB(w_1, c)$ and $\NEB(w_2, c)$ (lines 5--10). For each
candidate $c \in NW$, we collect the easiest two $\NEB$ assertions which verify
that $c$ cannot win into $\NWA$. Any alternate winner pair that includes a
candidate $c \in NW$ can be immediately ruled out with the chosen $\NEB$
assertions in $\NWA$.

\begin{figure}[t]
\begin{tabbing}
xx \= xx \= xx \= xx \= xx \= xx \= \kill
\textsf{NonWinners}() \\
1\> $AG$ := $NW$ := $\NWA$ := $\emptyset$ \\
2\> \textbf{forall} $w \in \cand$, $l \in \cand - \{w\}$\\
3\> \> \textbf{if} $\NEB(w,l)$ \\
4\> \> \> $AG$ := $AG \cup \{(w,l)\}$ \\
5\> \textbf{forall} $c \in \cand$ \\
6\> \> \textbf{if} $|\{ w : (w,c) \in AG\}| \geqslant 2$ \\
7\> \> \> $NW$ := $NW  \cup \{c\}$ \\
8\> \> \> $w1$ := $\argmin_w \{ ASN(\NEB(w,c),\alpha,\epsilon) : (w,c) \in AG \}$ \\
9\> \> \> $w2$ := $\argmin_w \{ ASN(\NEB(w,c),\alpha,\epsilon) : (w,c) \in AG, w \neq w1 \}$ \\
10\> \> \> $\NWA$ := $\NWA \cup \{\NEB(w1,c), \NEB(w2,c) \}$ \\ 
11\> \textbf{return} ($AG$, $NW$, $\NWA$) 
\end{tabbing}
\caption{Pseudo-code to calculate definite non-winners $c$, for which we have
at least two candidates where $\NEB(w1,c)$ and $\NEB(w2,c)$ hold.  The function
returns the set $AG$ of always greater relations, the set $NW$ of non-winners,
and the set of assertions $\NWA$ required to verify this.}
\label{fig:nw-short}
\end{figure}

\begin{figure}
\begin{tabbing}
xx \= xx \= xx \= xx \= xx \= xx \= \kill
\textsf{FindAuditableAssertions}() \\
1\> $(AG,NW,\NWA)$ := \textsf{NonWinners}() \\
2\> $\assertion$ := $\NWA$ \\
3\> $ASN$ := $\max \{ ASN(a,\alpha,\epsilon) : a \in \assertion \}$ \\
4\> \textbf{forall} $\{c_1,c_2\} \subset \cand, \{c_1,c_2\} \neq \{w_1,w_2\}$ \\
 \> \> \% for each pair, find the easiest way to eliminate it \\
5\> \> \textbf{if} ($\{c_1,c_2\} \cap NW \neq \emptyset$) \textbf{continue} \\
6\> \> $LASN$ := $+\infty$ \\
7\> \> $G_1$ := $\{ g ~:~ a \in \cand - W, (g,c_1) \in AG \}$ \\
8\> \> $G_2$ := $\{ g ~:~ a \in \cand - W, (g,c_2) \in AG \}$ \\
9\> \> \textbf{forall} $o \in \cand - \{c_1,c_2\}$ \\
10\> \> \> $O$ := $\{ o' : (o,o') \in AG \}$ \\
 \> \> \> \% assume $c_1$ wins, show $c_2$ is eliminated \\
11\> \> \> \textbf{if} $\CNEB(o,c_2,\{c_1\},\{2/3\},G_2,O - \{c_1\})$ holds \\
12\> \> \> \> $LA'$ = $\{ \CNEB(o,c_2,\{c_1\},\{2/3\},G_2,O - \{c_1\})\} \cup \{ \NEB(o,o') : o' \in O \}\cup \{ \NEB(g,c_2) : g \in G_2 \}$ \\
13\> \> \> \> $LASN'$ := $\max \{ ASN(a,\alpha,\epsilon) :  a \in LA'\}$ \\ 
14\> \> \> \> \textbf{if} $LASN' < LASN$ \\
15\> \> \> \> \> $LASN$ := $LASN'$ \\ 
16\> \> \> \> \> $LA$ := $LA'$ \\
  \> \> \> \% assume $c_2$ wins, show $c_1$ is eliminated \\
17\> \> \> \textbf{if} $\CNEB(o,c_1,\{c_2\},\{2/3\},G_1,O - \{c_2\})$ holds \\
18\> \> \> \> $LA'$ = $\{ \CNEB(o,c_1,\{c_2\},\{2/3\},G_1,O - \{c_2\})\} \cup \{ \NEB(o,o') : o' \in O \} \cup \{ \NEB(g,c_1) : g \in G_1 \}$ \\
19\> \> \> \> $LASN'$ := $\max \{ ASN(a,\alpha,\epsilon) : a \in LA' \}$ \\
20\> \> \> \> \textbf{if} $LASN' < LASN$ \\
21\> \> \> \> \> $LASN$ := $LASN'$ \\ 
22\> \> \> \> \> $LA$ := $LA'$ \\
23\> \> \textbf{if} $LASN = +\infty$ \\
24\> \> \> \textbf{abort} \% no auditable assertions \\ 
25\> \> $ASN$ := $\max(ASN,LASN)$ \\
26\> \> $\assertion$ := $\assertion \cup LA$ \\
27\> \textbf{return} $\assertion$
\end{tabbing}
\caption{Calculate a set of assertions $\assertion$ sufficient to verify a
2-seat STV election.}
\label{fig:audit}
\end{figure}

Once we have a reduced set of alternate winner pairs, we use
\textsf{FindAuditableAssertions} (\autoref{fig:audit}) to find more complex
$\CNEB$ assertions that would rule them out.  The initial set of assertions is
set to $\NWA$, as produced by \textsf{NonWinners}. The current expected ASN is
given by $ASN$, and this will increase over the course of the algorithm.  We
then consider every alternate pair of winners, excluding those involving a
candidate in $NW$, and find a set of assertions $LA$, with an expected audit
cost of $LASN$, to eliminate this possibility.

For a given alternate winner pair $\{c_1,c_2\}$, we can rule out this outcome
by finding a candidate $o \in \cand - \{c_1,c_2\}$ for which we can show that
either: $o$ cannot be eliminated before $c_2$ in the context where $c_1$ is
seated (at some point); or, similarly, $o$ cannot be eliminated before $c_1$ in
the context where $c_2$ is seated.

We consider each candidate $o \in \cand - \{c_1,c_2\}$ in turn.  We first
consider if we can form an $\CNEB$ assertion showing that $o$ cannot be
eliminated before $c_2$ (or $c_1$) in the context where $c_1$ (or $c_2$) is
seated. We only need one such $\CNEB$ assertion to rule out the alternate
winner pair. As we consider each $o$, and these two different contexts, we keep
track of the easiest of these potential $\CNEB$ assertions. As described
earlier, an $\CNEB$ assertion between two candidates $w$ and $l$ will use a
number of pre-computed $\NEB$ assertions to guide which ballots should
contribute to a lower bound on the tally of $w$ and an upper bound on the tally
of $l$. When choosing a given $\CNEB$ to rule out the alternate winner pair
$\{c_1,c_2\}$, the set $LA$ contains the $\CNEB$ assertion and all $\NEB$
assertions that it uses (lines 12 and 18).

If we never find a way to eliminate a pair $\{c_1,c_2\}$ then the election is
not auditable with this approach; abort. Otherwise, update the global ASN, and
add the best assertions for removing $\{c_1,c_2\}$ to $\assertion$. Finally,
return $\assertion$.

\begin{theorem}
The set of assertions $\assertion$ returned by \textsf{FindAuditableAssertions}
(\autoref{fig:audit}) is sufficient to rule out all alternate election
results.
\end{theorem}
\begin{proof}
Each alternate election result is ruled out by $\assertion$.  For pairs where
$\{c_1,c_2\} \cap NW \neq \emptyset$, assume w.l.o.g.\ $c_1 \in NW$. Then by
\autoref{coro:NEB} there are two other candidates that will be seated if $c_1$
is seated, which rules out the pair.  Otherwise \autoref{lemma:wlo} shows that
one of $c_1$ or $c_2$ cannot win before another candidate $o$.
\qed
\end{proof}

\subsection{Two initial quotas case}\label{sec:TwoQuota}

The general algorithm described above can be applied to all 2-seat STV
elections but there are alternatives for some elections which might be easier
to audit.

Suppose $\IQ(w_1)$ and $\IQ(w_2)$ hold.  That is, both reported winners
achieved a quota initially. We can simply define $\assertion = \{\IQ(w_1),
\IQ(w_2)\}$ with an expected ASN of $\max\{ ASN(\IQ(w_1), \alpha, \epsilon),
ASN(\IQ(w_2), \alpha, \epsilon)\}$.

\subsection{One initial quota case} \label{sec:OneQuota}

A frequent occurrence in STV elections is that one candidate has a first
preference tally that exceeds a quota. We may be able to use this outcome
structure to generate a set of assertions that are easier to audit than those
found using the general algorithm. To generate a set of assertions to audit
such a 2-seat STV election, we start with the assertion $\IQ(w_1)$ for the
first seated candidate, $w_1$.

For the second reported winner, $w_2$, we then assert $\CNEB(w_2, c, \{w_1\},
\overline{\tau}, G, O)$ for all candidates $c \in \cand - \{w_1,w_2\}$ given an
assumed upper bound $\overline{\tau}$ on the transfer value of ballots leaving
$w_1$'s pile and a set of candidates $o \in O$ for which $\NEB(w_2, o)$ holds,
and $g \in G$ where $\NEB(g,c)$ holds.

For any choice of $\overline{\tau}$, we need to validate that the actual
transfer value for ballots leaving $w_1$'s tally is indeed below
$\overline{\tau}$. We do this with the assertion $\UT(w_1,\overline{\tau})$.

We could set $\overline{\tau}$ to the reported transfer value, $\tau_{w_1}$,
however the $\UT$ assertion would then have a zero margin and thus be
impossible to audit. The higher we set $\overline{\tau}$, up to a maximum value
of 2/3, the easier it will be to audit. However, as $\overline{\tau}$
increases, the $\CNEB$ assertions formed above (to check that $w_2$ cannot lose
to any reported losing candidate) become harder to audit. This is because the
upper bounds (on the tallies of these reported losers) in the context of each
$\CNEB$ increase as $\overline{\tau}$ increases. With this increase, the margin
of any $\CNEB$ that we can form decreases. 

To find an appropriate value of $\overline{\tau}$, we initialise the upper
bound to $\tau_{w_1}$ and gradually increase it in small increments, $\delta$.
For each choice of $\overline{\tau}$, we compute the set of $\CNEB$ assertions
required to show that $w_2$ cannot lose to any remaining candidate, keeping
track of the ASN required for that audit configuration. We continue to increase
$\overline{\tau}$ while the ASN of the resulting audit decreases. Once
$\overline{\tau}$ reaches 2/3, or the ASN of the audit configuration starts to
increase, we stop and accept the least-cost audit configuration found.

% ---------------------------------------------------------------------------

\section{Experimental results}

We ran the general, one-quota, and two-quota audit generation methods described
in \autoref{sec:Methods}, on a range of election instances: four 2-seat STV
elections conducted as part of the 2016 and 2019 Australian Senate elections;
and several US and Australian IRV elections re-imagined as 2-seat STV
elections.\footnote{Our code is publicly available at:
\url{https://github.com/michelleblom/stv-rla}} We used $\delta = 0.01$ for the
one-quota method, and all methods were implemented as ballot-comparison audits.
The ASNs for the resulting audits, based on a risk limit of 10\% and assumed
error rate of 0.2\%, are reported in \autoref{tab:Results}. A `--' indicates
that the given audit generation method was not applicable to the instance,
while $+\infty$ indicates that the method could not find an auditable set of
assertions.  Bold entries are instances where the general method is expected to
be more efficient than the one- and two-quota methods.

In general, the one-quota method formed the cheapest audit, where applicable.
This is expected to be case as the general approach assumes the highest
possible transfer value for ballots leaving the first winner's pile. The
one-quota method, in contrast, finds a trade-off between the difficulty of
checking that the transfer value for the first winner is less than an assumed
upper bound, and the difficulty of $\CNEB$ assertions to check that the second
winner cannot lose to any reported loser. The former is easier with a larger
upper bound, while the latter are easier with a smaller lower bound.

For the instances considered, the two-quota method forms more costly audits
than the one-quota and general methods. In instances where there is one
dominant candidate that receives significantly more first preference votes than
others, the second winner typically has a much smaller surplus. The size of
this surplus determines the margin of the assertion used to check that the
second winner is seated in the first round.

An advantage of the one-quota method is that we can form more $\NEB$s by using
the fact that we have a tight upper bound on the transfer value of ballots
leaving the first winner, who we know is seated in the first round. We can
create more of these $\NEB$s than would be possible if we were assuming an
upper bound of 2/3 on transfer values. With more available $\NEB$s, we can more
effectively increase and decrease the bounds on the tallies of candidates
within $\CNEB$s. This allows us to create more $\CNEB$s, including some that we
cannot form under the general method.

\begin{table}[t]
\centering
\caption{ASNs for audit configurations generated for four Australian Senate
    2-seat STV elections (2016 and 2019), and several US and Australian IRV
    elections re-imagined as 2-seat STV elections. We report the ASNs of audits
    generated using the one-quota, two-quota and general methods, where
    applicable. A risk limit of $\alpha = 10\%$ and error rate of $\epsilon =
    0.2\%$ were used.}
\scalebox{0.98}{
\begin{tabular}{|l|r|r|r|r|r|r|}
    \hline
    \textbf{Election} & $|\mathcal{C}|$ & \textbf{Valid} & \textbf{Quota} &
    \textbf{2-quota}  & \textbf{1-quota} & \textbf{General} \\
    & & \textbf{Ballots} & & \textbf{ASN} & \textbf{ASN} & \textbf{ASN} \\
    \hline
    2016 ACT & 22 & 254,767 & 84,923 & -- & 66 &  $+\infty$ \\
    2019 ACT & 17 & 270,231 & 90,078 & -- & 107 & $+\infty$ \\
    \hline
    2016 NT & 19 & 102,027 & 34,010 & 100 & 74 & 569 \\
    2019 NT & 18 & 105,027 & 35,010 & 100 & 72 & 327 \\
    \hline
    \hline
    \multicolumn{7}{|l|}{\textbf{IRV elections re-imagined as 2-seat STV elections}} \\
    \hline
    \multicolumn{7}{|l|}{\textit{\textbf{No candidate achieves a quota on first preferences}}} \\
    \hline
NSW'19 Barwon & 9 & 46,174 & 15,392 & -- & -- & 285 \\
2014 Oakland Mayor & 17 & 101,431 & 33,811 & -- & --& $+\infty$ \\
2014 Berkeley City Council D8 & 5 & 4,497 & 1,500 & -- & -- & $+\infty$ \\
2009 Aspen City Council & 11 & 2,487 & 830 & -- & -- & $+\infty$\\
2008 Pierce CAS & 7 & 262,447 & 87,483 & -- & -- & $+\infty$ \\
\hline
    \hline
    \multicolumn{7}{|l|}{\textit{\textbf{At least one candidate achieves a quota on first preferences}}} \\
    \hline
\multicolumn{7}{|l|}{\textbf{US elections}} \\
\hline
2013 ward 5 & 5 & 3,499 & 1,167 & -- & 114 & $+\infty$\\
OK CC D2 2014 & 6 & 13,500 & 4,501 & $+\infty$ & 100 & 127\\
Aspen 2009 Mayor & 5 & 2,528 & 843 &203 & 51 & 195 \\
2010 Berkeley CC D1 & 5 & 5,700 & 1,901 & -- &69 & 921 \\
2010 Berkeley CC D7 & 4 & 4,184 & 1,395 & 267 & 48 & 89 \\
Oakland 2010 Mayor &  11 &  119,607 &  39,870 &-- & 1,177  & $+\infty$ \\
Oakland 2010 CC D6 &  4 & 12,911 & 4,304 & -- & $+\infty$ &  $+\infty$\\
Pierce 2008 CA & 4 & 153,528 & 51,177 & 34 &19 & 23\\
Pierce 2008 CE & 5 & 299,132 & 99,711 &-- & 192 & $+\infty$\\
San Leandro 2010 D5 CC & 7 & 22,484  & 7,495 & 149 & 126 & 230\\
\hline
\multicolumn{7}{|l|}{\textbf{Australian elections: NSW 2019 Legislative Assembly}} \\
\hline
Auburn & 5 & 44,842 & 14,948 & 107 & 25  & 46\\
Bathurst &  6 & 50,833 & 16,945 & -- & 95 &  125\\
Blue Mountains & 7 & 49,228 & 16,410 & -- & 71  & $+\infty$\\
Clarence& 6& 49,355& 16,452 & -- & 116 & 251\\
Granville & 8 & 44,191 & 14,731 & 67 &19 & 29\\
Hornsby & 9 & 50,003 & 16,668 &-- & 123 & 2,210\\
\bf Kogarah & \bf 5 & \bf 45,576 & \bf 15,193 & \bf 30 & \bf 30 & \bf 20\\
 Ku-ring-gai &  6 &  48,730 &  16,244 &-- &  229 &  340\\
Lakemba & 6 & 44,615 & 14,872 &-- & 60 & 607\\
Macquarie Fields & 6 & 52,789 & 17,597 &-- & 29 & 39\\
\bf Murray & \bf 10 & \bf 47,233 & \bf 15,745 & \bf 145 & \bf 50 & \bf 25\\
Myall Lakes & 6 & 50,315 & 16,772 &-- & 28 & 41\\
Newcastle & 8 & 50,319 & 16,774 &-- & 211 & $+\infty$\\
Newtown & 7 & 46,312 & 15,438 &-- & 49 & 67\\
North Shore & 9 & 47,774 & 15,925 &-- &  509 & 1200\\
Northern Tablelands & 4 & 48,678 & 16,227 & -- & $+\infty$ & $+\infty$\\
Oatley & 5 & 48,120 & 16,041 & $+\infty$& 21 & 23\\
Parramatta & 7 & 48,728 & 16,243 & -- &26 & 31\\
\bf Penrith & \bf 10 & \bf 48,853 & \bf 16,285 &\bf 118 & \bf 40 & \bf 24\\
Pittwater & 8 &  49,119 &  16,374 &  -- & 190 &  223\\
Port Macquarie & 4 & 52,735 & 17,579 &-- & 44 & 57\\
Riverstone & 3 & 53,510 & 17,837 & 41 &16 & 18\\
\bf Upper Hunter & \bf 8 & \bf 48,525 & \bf 16,176 &-- & \bf 520 & \bf 145\\
Vaucluse & 7 & 46,023 &  15,342 &-- & 508 & $+\infty$\\
\hline
\end{tabular}}
\label{tab:Results}
\end{table}

% ---------------------------------------------------------------------------

\section{Conclusion}

We have presented the first method we are aware of for risk-limiting audits for
STV elections, restricted for the moment to 2-seat STV elections.

We were able to design an efficient audit for all of the \emph{real-world}
2-seat STV elections for which we have data, although the general method is not
strong enough for two of them.  For other elections---where we re-imagine IRV
elections as 2-seat STV elections---we see that if no candidate has a quota
initially, we struggle to find an auditable set of assertions. In the case that
one or two candidates initially obtain a quota, we are usually able to audit
them successfully, with the one-quota method usually requiring less effort, but
not always.

The assertions we define in this paper are not specific to 2-seat STV
elections. Thus, they provide a starting point for auditing STV elections with
more seats. Obviously even the 2-seat case is not easy, so investigating
tighter lower and upper bounds on tallies is likely to be valuable.

% ---------------------------------------------------------------------------

\bibliographystyle{splncs04}
\bibliography{references}

% ---------------------------------------------------------------------------

\end{document}